\documentclass{birkjour}

\usepackage{latexsym}
\usepackage{graphics}
\usepackage{hyperref, enumerate}
\usepackage{soul}
\usepackage{graphicx}
\usepackage{amsmath}
\usepackage{amssymb}
\usepackage{verbatim}
\usepackage{color}
\usepackage{cancel}
\usepackage{nicefrac}

\def\calh{{\mathcal H}}
\def\C{\mathbb C}
\def\norm#1{\left\Vert #1\right\Vert }
\def\Tr{\mathop {\rm Tr}}

\def\anti#1#2{\left\{#1,#2\right\}}
\def\comm#1#2{\left[#1,#2\right]}

\def\be{\begin{equation}}
\def\beq{\begin{eqnarray}}
\def\ee{\end{equation}}
\def\eeq{\end{eqnarray}}
\def\eqref#1{(\ref{#1})}
\def\lra#1{\langle #1 \rangle}
\def\lrp#1{\left( #1 \right)}
\def\rp{\text{RP}}
\def\abs#1{\left\vert #1\right\vert}
\def\B{\mathcal{B}}
\def\A{\mathfrak{A}}
 
\def\calh{\mathcal{H}}

\def\rp{\text{RP}}
\def\even{{\rm even}}

\renewcommand{\thesection}{\Roman{section}}
\renewcommand{\thesubsection}{\thesection.\arabic{subsection}}
\renewcommand{\theequation}{\thesection.\arabic{equation}}

\newtheorem{thm}{Theorem}
\newtheorem{cor}[thm]{Corollary}
\newtheorem{lem}[thm]{Lemma}
\newtheorem{prop}[thm]{Proposition}

\usepackage{graphics}

\begin{document}

\title{Reflection Positivity for Majoranas}

\author[Arthur Jaffe]{Arthur Jaffe}
\address{Harvard University, Cambridge, Massachusetts 02138, USA\\ Department of Physics, University of Basel, Basel, Switzerland\\
Institute for Theoretical Physics, ETH Z\"urich, Z\"urich Switzerland}
\email{Arthur{\_}Jaffe@harvard.edu}

\author[Fabio L. Pedrocchi]{Fabio L. Pedrocchi}
\address{Department of Physics, University of Basel, Basel, Switzerland}
\email{fabio.pedrocchi@unibas.ch}

\maketitle

\begin{abstract}
We establish reflection positivity for Gibbs trace states defined by a certain class of Hamiltonians that describe the interaction of Majoranas on a lattice. These Hamiltonians may include many-body interactions, as long as the signs of the associated coupling constants satisfy certain restrictions. We show that reflection positivity holds on an even subalgebra of Majoranas. 
\end{abstract}

\setcounter{equation}{0}  \section{Introduction}
In this paper we prove reflection positivity for trace functionals defined by a certain class of interactions of (neutral) Majoranas on a lattice.  Earlier results on reflection positivity for fermions in the framework of quantum statistical mechanics focus on the case of  charged excitations.  In  \S\ref{sect:Hamiltonians} we isolate conditions  that entail reflection positivity on an interaction Hamiltonian $H$, expressed in terms of Majoranas. Our main result is Theorem~\ref{prop:reflection_positivity} of~\S\ref{sect:MainResult},
\begin{equation}
0\leqslant\Tr(A\,\vartheta(A)\,e^{-H})\,,
\end{equation}
which is valid for certain functions $A$ of Majoranas, and for a reflection $\vartheta$.
Some related bounds are given in \S\ref{sect:Reflection_Bounds}.

Our formulation and  proof of Theorem \ref{prop:reflection_positivity} in \S\ref{sect:MainResult} involve familiar methods, but they also require new ideas. As the present paper describes interactions without charge,  one does not have the useful charge-conservation symmetry to aid in their analysis. In this case we establish reflection positivity on an even sub-algebra of fermions. The corresponding positivity is not valid on the full fermionic algebra for a half-space on one side of the reflection plane, as we show with an explicit counterexample in~\eqref{eq:Counterexample}.

Recently the present authors have studied certain  quantum spin interactions, which are of interest in quantum information theory \cite{CJLP}, where we apply the reflection positivity results of this present paper. These quantum spin systems have certain features similar to lattice gauge theory. However one must also deal with the additional complication that the basic fermionic variables anti-commute at different sites, rather than commute.   

Reflection positivity has played an important role in analysis of quantum fields as well as the analysis of classical and quantum spin systems. 
Osterwalder and Schrader  discovered reflection positivity in their study of classical fields on Euclidean space \cite{OS1}; it provided the key notion of quantization and allowed one to go from a classical field to a  quantum-mechanical Hilbert space and a positive Hamiltonian acting on that Hilbert space. 

Multiple reflection bounds, based on reflection positivity for classical fields, played a crucial role in Glimm, Jaffe, and Spencer's mathematical proof  \cite{GJS} of the physicists' assumption that phase transitions and symmetry breaking exist in quantum field theory. This first example of a phase transition in field theory \cite{GJS} concerned breaking of a discrete $\mathbb{Z}_{2}$ symmetry. 
Reflection positivity also turned out be be extremely useful in the analysis of lattice models for boson and fermion interactions by Fr\"ohlich, Simon, Spencer, Dyson, Israel, Lieb, Macris, Nachtergale, and others \cite{FSS,DLS,FILS,Lieb1994,Nachtergaele1996}.  This included the analysis of phase transitions and the breaking of certain continuous symmetry groups in lattice spin systems. In addition, reflection positivity was crucial in the study by Osterwalder and Seiler  of  the Wilson action for lattice gauge theory \cite{OS3}. 

\setcounter{equation}{0}  \section{Definitions and Basic Properties}
Majoranas on a lattice are a self-adjoint representation of a Clifford algebra with $2N$ generators $c_{i}$. They satisfy
\be
\anti{c_i}{c_j}=2\delta_{ij}\,,\qquad c_{i}^{*}=c_{i}\,,\qquad \text{for } i,j=1,\ldots,2N\,.
\ee  
One can realize $2N$ Majoranas in a standard way on a complex Hilbert space of dimension $2^{N}$, and we use this representation.  Start with the real Hilbert space  $\mathcal{H}_{r}=\wedge \mathbb{R}^{N}$,  the real exterior algebra over  $\mathbb{R}^{N}$. Let $a_{j}^{*}$ denote the linear transformation  on $\mathcal{H}_{r}$ given  the exterior product  $e_{j}\wedge$ with the $j^{\rm th}$ basis element $e_{j}$ in $\mathbb{R}^{N}$. These operators and their adjoint $a_{j}$ are $N$ fermionic creation and annihilation operators.  Let $\calh$ denote the complexification of $\mathcal{H}_{r}$ and define   the Majorana operators $c_{2j-1}, c_{2j}$ 
as linear combinations,  $c_{2j-1}=a_{j}+a_{j}^{*}$ and $c_{2j}=i\lrp{a_{j}-a_{j}^{*}}$. Thus our odd indexed Majoranas are real and the even Majoranas are purely imaginary.

We consider the index $j$ of the Majoranas to have a geometric significance as an element of a simple cubic lattice $\Lambda=\Lambda_{-}\cup\Lambda_{+}$. We assume that $\Lambda$ is invariant under a reflection $\vartheta$ in a plane $\Pi$ normal to a coordinate direction and intersecting no sites in $\Lambda$, so $\vartheta(\Lambda)=\Lambda$. Here $\Lambda_{\pm}$ denote the sites on the $\pm$ side of $\Pi$. We assume that the reflection $\vartheta$ maps $\Lambda_{\pm}$ into $\Lambda_{\mp}$. 

For any subset $\B\subset\Lambda$, let  $\A(\mathcal B)$ denote the algebra generated by the $c_{j}$'s with $j\in\B$.  Let $\A=\A(\Lambda)$ and $\A_{\pm}=\A(\Lambda_{\pm})$.  Also introduce the even algebras $\A(\B)^{\even}$, as the subset of $\A(\B)$ generated by even monomials in the $c_{j}$'s, with $j\in\B$. Note that $\A^{\even}$ is not abelian, but $\A^{\even}(\B)$ commutes with $\A^{\even}(\mathcal{B'})$ when $\B\cap\mathcal{B'}=\varnothing$.

\subsection{Anti-Unitary Transformations}
An antilinear transformation $\Theta$ on the finite-dimensional complex Hilbert space $\mathcal{H}$ has the property $\Theta(f+\lambda g)=\Theta f + \bar{\lambda} \Theta g$ for $f,g\in\mathcal{H}$ and $\lambda\in\mathbb{C}$. Here $\bar\lambda$ denotes the complex conjugate of $\lambda$. Assuming $\mathcal{H}$ has the hermitian inner product $\lra{\ \cdot \, , \cdot \ }$, the adjoint $\Theta^{*}$ of $\Theta$  is the anti-linear transformation
\be
\lra{f,\Theta^{*}g}=\lra{g,\Theta f}\,.
\ee
Also $\Theta$ is said to be anti-unitary if for all $f,g\in\mathcal{H}$,
\be
\lra{f,g}=\lra{\Theta g,\Theta f}=\lra{\Theta^{*}g,\Theta^{*}f}\,.
\ee
As a consequence an anti-unitary satisfies $\Theta\Theta^{*}=\Theta^{*}\Theta=I$ or $\Theta^{*}=\Theta^{-1}$. 

We are especially interested in an anti-unitary representation of the reflection $\vartheta$ on $\mathcal{H}$, which we also denote by $\vartheta$.  The anti-unitary $\vartheta$ defines an anti-linear map on $\A$, with $\vartheta:\A_{\pm}\rightarrow\A_{\mp}$ 
with the property 
\be
\vartheta(c_{j})
=\vartheta c_{j} \vartheta^{-1}
=c_{\vartheta j}\;.
\ee
By the general properties of the anti-unitary $\vartheta$, 
	\be
		\vartheta(AB)=\vartheta(A)\,\vartheta(B)\,,
		\qquad\text{and}\quad
		\vartheta(A)^{*} = \vartheta(A^{*})\;.
	\label{eq:Anti-homomorphism Theta}
	\ee
In addition
	\be
		\Tr(\vartheta(A)) = \overline{\Tr(A)}\;,
		\qquad\text{for all }
		A\in \mathfrak{A}\;.
	\label{eq:Reflected Trace}
	\ee
Thus  the Clifford algebra relations are also satisfied by $\vartheta(c_{j})$,  
	\be
		\anti{\vartheta(c_{i})}{\vartheta(c_{j})}=2\delta_{ij}I\;.
	\ee
It is no complication to allow a set of $n$ Majorana operators at each lattice site~$i$.   
 
\setcounter{equation}{0}  \setcounter{equation}{0}  \section{Hamiltonians}\label{sect:Hamiltonians}
We consider self-adjoint Hamiltonians of the form
	\be	
	\label{eq:Hamiltonian}
		H=H_{-}+H_{0}+H_{+}\,,
	\ee
where $H_{-}=H_{-}^{*}\in\mathfrak{A}_{-}^{\even}$ and $H_{+}=H_{+}^{*}\in\mathfrak{A}_{+}^{\even}$. The operator $H_{0}=H_{0}^{*}$ denotes a coupling across the reflection plane $\Pi$. Let $\mathfrak{I}=\{i_{1},\ldots,i_{k}\}$ denote a subset of points in $\Lambda_{-}$ with cardinality $n(\mathfrak{I})=\abs{\mathfrak{I}}$. Define
	\be
	\label{eq:sigma defn}
		\sigma(\mathfrak{I})=n(\mathfrak{I})\mod2\,.
	\ee
We assume that $H_{0}$ has the form   
	\be
	\label{eq:Hamiltonian_interaction}
		H_{0}
		=\sum_{\mathfrak{I}}J_{\mathfrak{I}\,\vartheta \mathfrak{I}}\
			 i^{\sigma(\mathfrak{I})}\,C_{\mathfrak{I}}\,\vartheta(C_{\mathfrak{I}})\,,
			 \qquad\text{where}\quad
			 J_{\mathfrak{I}\,\vartheta \mathfrak{I}}\in\mathbb{R}\;, 
	\ee
and $C_{\mathfrak{I}}=c_{i_{1} }  c_{i_{2}} \cdots c_{i_{k}}\in\mathfrak{A}_{-}$.

\smallskip
\noindent\textbf{Remark:}
The Hamiltonian $H_{0}$ is self-adjoint and reflection-symmetric, 
	\be
		H_{0}
		=H_{0}^{*}=\vartheta(H_{0})\,.
	\ee
Each term in the sum  \eqref{eq:Hamiltonian_interaction}  defining $H_0$ is self-adjoint. In fact 
	\be
		\lrp{C_{\mathfrak{I}}\,\vartheta(C_{\mathfrak{I}})}^{*}
		=\vartheta(C_{\mathfrak{I}})^{*}\,C_{\mathfrak{I}}^{*}
		=(-1)^{\abs{\mathfrak{I}}}\,C_{\mathfrak{I}}\,
			\vartheta(C_{\mathfrak{I}})\,.
	\ee
So from  $\overline{i^{\sigma(\mathfrak{I})}}=(-1)^{\sigma(\mathfrak{I})}\,
i^{\sigma(\mathfrak{I})}$, and $(-1)^{\sigma(\mathfrak{I})}=(-1)^{\abs{\mathfrak{I}}}$, we infer
	\be
	\lrp{i^{\sigma(\mathfrak{I})}\,C_{\mathfrak{I}}\,\vartheta(C_{\mathfrak{I}})}^{*}
	=i^{\sigma(\mathfrak{I})}\,C_{\mathfrak{I}}\,\vartheta(C_{\mathfrak{I}})\;.
	\ee	
Likewise 
	\be
		\vartheta(H_{0})
		=\sum_{\mathfrak{I}}(-1)^{\abs{\mathfrak{I}}}\,i^{\sigma(\mathfrak{I})}\,\vartheta(C_{\mathfrak{I}})\,C_{\mathfrak{I}}=\sum_{\mathfrak{I}}i^{\sigma(\mathfrak{I})}\,C_{\mathfrak{I}}\,\vartheta(C_{\mathfrak{I}})\,.
	\ee
Here we use the fact that the $\abs{\mathfrak{I}}$ Majoranas in $C_{\mathfrak{I}}$ all anti-commute with the ones in $\vartheta(C_{\mathfrak{I}})$, yielding another factor $(-1)^{\abs{\mathfrak{I}}}$ in the final equality.

\smallskip
\noindent\textbf{Assumptions on the Couplings:} We require that the sign of the couplings $J_{\mathfrak{I}\,\vartheta \mathfrak{I}}$ in \eqref{eq:Hamiltonian_interaction} satisfy 
	\be	\label{eq:CouplingRestriction}
	\begin{array}{lll}
			&\text{all } J_{\mathfrak{I}\,\vartheta \mathfrak{I}}\leqslant0\;,
			\,\text{or all }J_{\mathfrak{I}\,\vartheta \mathfrak{I}}\geqslant0\;,
			&\text{for terms with } \sigma(\mathfrak{I}) =1\,,\\
			&\text{all } J_{\mathfrak{I}\,\vartheta \mathfrak{I}} 	\leqslant0\,,
			&\text{for terms with } \sigma(\mathfrak{I}) =0\,.
	\end{array}
	\ee

We restrict the sign of couplings only for interaction terms \eqref{eq:Hamiltonian_interaction} that cross the plane $\Pi$. Nearest-neighbor two-body interactions have $\sigma(\mathfrak{I})=1$.

\setcounter{equation}{0}  \section{Monomial Basis}
The $2N$ operators $c_{i}$ yield monomials of the form $M_{\beta}=c_{i_{1}}c_{i_{2}}\cdots c_{i_{j}}$ of degree $j$, with $i_{1}<i_{2}<\cdots i_{j}$.  (Other orders of the $c$'s are the same up to  a $\pm$ sign.) Denote by $\beta=0$ the monomial $M_{0}=I$.    There are $2N\choose j$ such monomials $M_{\beta}$ of degree $j$, so there are a total of $2^{2N}$ such monomials. As $2^{2N}=\lrp{\dim \mathcal{H}}^{2}$, these monomials are a candidate for a basis of the space of matrices acting on $\mathcal{H}$. 

\begin{prop}\label{prop:monomials}
If $\beta\neq0$, the monomials $M_{\beta}$ have vanishing trace, $\Tr\lrp{M_{\beta}}=0$. Any linear transformation $A$ on $\mathcal{H}$ can be written in terms of the basis $M_{\beta}$ as
\be\label{eq:A_expansion}
A=\sum_{\beta}a_{\beta}\,M_{\beta}\,,\quad\text{where}\quad a_{\beta}=2^{-N}\Tr\left(M_{\beta}^{*}A\right)\,.
\ee
The monomials $M_{\beta}$ are an irreducible set of matrices. 
\end{prop}
\begin{proof}
If $\deg M_{\beta}$ is odd, there is at least one of the $c$'s, say $c_{j}$, not contained in $M_{\beta}$. Thus 
	\be
		\Tr\lrp{M_{\beta}}
		= \Tr\lrp{c_{j} c_{j}M_{\beta}}
		= \Tr\lrp{ c_{j}M_{\beta}c_{j}}
		= (-1)^{\deg M_{\beta}}\Tr\lrp{M_{\beta}}
		= - \Tr\lrp{M_{\beta}}
		=0\;.\nonumber
	\ee
On the other hand, if $\deg M_{\beta}=2k>0$, and $c_{j}$ does occur in $M_{\beta}$, then also 
	\be
		\Tr\lrp{M_{\beta}}
		= \Tr\lrp{c_{j}^{2}M_{\beta}}
		= \Tr\lrp{ c_{j}M_{\beta}c_{j}}
		= (-1)^{\deg M_{\beta}-1}\Tr\lrp{M_{\beta}}
		= - \Tr\lrp{M_{\beta}}
		=0\;.\nonumber
	\ee
Thus we have verified the first statement in the proposition. Also $M_{\beta}^{*}M_{\beta}=I$, and $M_{\beta'}^{*}M_{\beta}=\pm M_{\gamma}$ for some $\gamma\neq0$.

Suppose that there are coefficients $a_{\beta}\in\mathbb{C}$ such that $\sum_{\beta} a_{\beta}M_{\beta}=0$.  Then for any $\beta'$, one has ${M_{\beta'}^{*} \sum_{\beta} a_{\beta}M_{\beta}}= \sum_{\beta} a_{\beta} {M_{\beta'}^{*}M_{\beta}}=0$. Taking the trace shows that $a_{\beta'}=0$, so the $M_{\beta}$ are actually linear independent. As there are $2^{2N}$ matrices $M_{\beta}$, they are a basis for all matrices on $\mathcal{H}$.

Expanding an arbitrary matrix $A$ in this basis, we calculate the coefficients in \eqref{eq:A_expansion} using $\Tr I=2^{N}$. As the set of all matrices on $\mathcal{H}$ is irreducible, the basis $M_{\beta}$ is also irreducible. 
\end{proof}
\goodbreak

\setcounter{equation}{0}  \section{Reflection Positivity}\label{sec:RP}
In this section we consider traces on the Hilbert space $\calh = \wedge \C^{N}$.  

\begin{prop}[\bf Reflection Positivity I]\label{prop:positivity}
Consider an operator $A\in\mathfrak{A}_{\pm}$, then
\be
{\Tr}(A\,\vartheta(A))\geqslant0\,.
\ee
\end{prop}
\begin{proof}
The operator $A\in\mathfrak{A}_{\pm}$ can be expanded as a polynomial in the basis $M_{\beta}$ of Proposition \ref{prop:monomials}. The monomials that appear in the expansion all belong to $\mathfrak{A}_{\pm}$.   Write
\be
A=\sum_{\beta}a_{\beta}\,M_{\beta}\;,\qquad
\text{and}\quad
\vartheta(A)=\sum_{\beta} \overline{a_{\beta}}\,\vartheta(M_{\beta})\;.
\ee
We now consider the case $A\in\mathfrak{A}_{-}$.
For $M_{ \beta}=c_{i_{1}}\cdots c_{i_{k}}$, define
$
			M_{\vartheta\beta}
			=c_{\vartheta i_{1}}\cdots c_{\vartheta i_{k}}\,.$
 One then has 
	\be
		\Tr\lrp{A\,\vartheta(A)}
		=
		\sum_{\beta,\beta^{'}}
		a_{\beta} \,\overline {a_{\beta'}}\, \Tr\lrp{M_{\beta}\,\vartheta(M_{\beta'})}
		=\sum_{\beta,\beta^{'}}
		a_{\beta} \,\overline {a_{\beta'}}\, \Tr\lrp{M_{\beta} \, M_{\vartheta \beta'}}\,.
	\ee
Since $M_{\beta}\in\mathfrak{A}_{-}$ and $M_{\vartheta\beta'}\in\mathfrak{A}_{+}$, they are products of different Majoranas. We infer from Proposition \ref{prop:monomials} that the trace vanishes unless $\beta=\vartheta\beta'=0$. We have,  
\be
\text{Tr}\left(A\,\vartheta(A)\right)
=
2^{N}
\left\vert a_{0}\right\vert^2\geqslant0\,,
\ee
as claimed.
\end{proof}

This reflection positivity allows one to define a pre-inner product on $\mathfrak{A}_{\pm}$ given by
\be\label{eq:inner_product_1}
\lra{A,B}_{\rp}=\Tr(A\,\vartheta(B))\,.
\ee
This pre-inner product satisfies the Schwarz inequality
\be\label{eq:Schwarz}
\abs{\lra{A,B}_{\rp}}^2\leqslant \lra{A,A}_{\rp}\,\lra{B,B}_{\rp}\,.
\ee
In the standard way, one obtains an inner product $\lra{\widehat{A},\widehat{B}}_{\rp}$ and norm $\Vert\widehat{A}\Vert_{\rp}$ by defining the inner product on equivalence classes $\widehat{A}=\{A+n\}$ of $A$'s, modulo elements $n$ of the null space of the functional \eqref{eq:inner_product_1} on the diagonal.   In order to simplify notation, we ignore this distinction.

\section{The Main Result\label{sect:MainResult}}
Here we consider reflection-positivity of the functional 
\be\label{eq:rp functional}
\Tr(A\,\vartheta(B)\,e^{-H})\,,
\qquad\text{for}\quad
A,B\in \mathfrak{A}_{\pm}^{\even}\;,
\ee
that is linear in $A$ and anti-linear in $B$.  

\begin{thm}[\bf Reflection Positivity II]\label{prop:reflection_positivity}
Consider $A\in\mathfrak{A}_{\pm}^{\even}$ and $H$ of the form \eqref{eq:Hamiltonian}, with $H_{+}=\vartheta(H_{-})$. Then the functional \eqref{eq:rp functional} is positive on the diagonal,
\be\label{eq:rp}
0\leqslant\Tr(A\,\vartheta(A)\,e^{-H})\,.
\ee
\end{thm}

\noindent\textbf{Remark:}
The functional \eqref{eq:rp} does not satisfy reflection positivity on the full fermonic algebra $\mathfrak{A}_{\pm}$. Even for $N=1$, with $H_{\pm}=0$, $H_{0}=-i\,c_{1}\vartheta(c_{1})$, and $A=c_{1}$, reflection positivity fails. In this case 
	\be\label{eq:Counterexample}
		\Tr(A\,\vartheta(A)\,e^{-H})=-2i\sinh 1\;,
	\ee
is purely imaginary. A similar argument shows that reflection positivity fails in case the coupling constants do not obey the restrictions \eqref{eq:CouplingRestriction}.

If the interaction terms in $H_{0}$ all have $\sigma_{\mathfrak{I}}=0$, then the functional \eqref{eq:rp} vanishes on odd elements of $\mathfrak{A}$,  and in this case reflection-positivity extends trivially to the full algebra.

There is a natural second reflection positivity  condition connected with the functional 
\be\label{eq:rp functional prime}
\Tr(\vartheta(A)B\,e^{-H})\,,
\qquad\text{for}\quad
A,B\in \mathfrak{A}_{\pm}^{\even}\;,
\ee
in place of \eqref{eq:rp}.  The properties \eqref{eq:Anti-homomorphism Theta}--\eqref{eq:Reflected Trace} ensure that 
	\be
		\Tr(\vartheta(A)B\,e^{-H})
		=\overline{\Tr(A\vartheta(B)\,e^{-\vartheta(H)})}\,.
	\ee
Since the assumed properties for $H$ hold also for $\vartheta(H)$ with $H_{\mp}$ replaced by $\vartheta(H_{\pm})$, we infer the following corollary.

\begin{cor}[\bf Reflection Positivity III]\label{prop:reflection_positivity cor}
Consider $A\in\mathfrak{A}_{\pm}^{\even}$ and $H$ of the form \eqref{eq:Hamiltonian}, with $H_{+}=\vartheta(H_{-})$. Then the functional \eqref{eq:rp functional prime} is positive on the diagonal,
\be\label{eq:rp prime}
0\leqslant\Tr(\vartheta(A)A\,e^{-H})\,.
\ee
\end{cor}

\begin{proof}[Proof of Theorem \ref{prop:reflection_positivity}]
Our argument is motivated by \cite{DLS,FILS,Lieb1994}, but has its own special features. Take  $A\in\mathfrak{A}_{-}^{\even}$.  Use the Lie product formula for matrices $\alpha_{1}$, $\alpha_{2}$, and $\alpha_{3}$ in the form 
\be\label{eq:LT}
 e^{\alpha_{1}+\alpha_{2}+\alpha_{3}}=\lim\limits_{k\rightarrow\infty}\left((1+\alpha_{1}/k)e^{\alpha_{2}/k}e^{\alpha_3/k}\right)^{k}\;.
\ee
This is norm-convergent for matrices. Take $\alpha_{1}=-H_{0}$, $\alpha_{2}=-H_{-}$, and $\alpha_{3}=-H_{+}=-\vartheta(H_{-})$ in \eqref{eq:LT}.

Label the non-empty subsets of $\Lambda_{-}$ by $\mathfrak{I}_{\ell}$, for $\ell=1,\ldots,L-1$, with $L= 2^{\abs{\Lambda_{-}}}$, and the empty subset $\varnothing$ by $\mathfrak{I}_{0}$. Let $H_{0}$ be defined in \eqref{eq:Hamiltonian_interaction}, with the sum ranging over the non-empty subsets. Write 
	\be
		H_{0}
		= 
		\sum_{\ell=1}^{L-1}{J_{\mathfrak{I}_{\ell}\,\vartheta \mathfrak{I}_{\ell}}}
		\ i^{\sigma(\mathfrak{I}_{\ell})}\,C_{\mathfrak{I}_{\ell}}\,\vartheta(C_{\mathfrak{I}_{\ell}})\;.
	\ee
Using \eqref{eq:LT},   
	\be\label{eq:Approx k}
		A\,\vartheta(A)\,e^{-H}
		= \lim\limits_{k\rightarrow\infty} 
			A\,\vartheta(A)\lrp{e^{-H}}_{k}\,,
	\ee
where
	\be\label{eq:Exp_minusH_k}
		\lrp{e^{-H}}_{k}
		=
		\lrp{(I
		-\sum_{\ell=1}^{L-1}{J_{\mathfrak{I}_{\ell}\,\vartheta \mathfrak{I}_{\ell}}}
		\ i^{\sigma(\mathfrak{I}_{\ell})}\,C_{\mathfrak{I}_{\ell}}\,\vartheta(C_{\mathfrak{I}_{\ell}})/k)
		\,e^{-H_{-}/k}\,e^{-\vartheta(H_{-})/k}}^{k}\;.
	\ee
One can include the term $I$ in the sums in \eqref{eq:Exp_minusH_k} by defining $-J_{\varnothing\,\vartheta\varnothing}=k$,  $C_{\varnothing}=C_{\vartheta\varnothing}=I$, and $n(\mathfrak{I}_{\ell_{0}})=n(\varnothing)=0$. Then 
	\beq\label{eq:Exp_minusH_k-2}
		\lrp{e^{-H}}_{k}
		&=&
		\frac{1}{k^{k}}
		\lrp{-\sum_{\ell=0}^{L-1}{J_{\mathfrak{I}_{\ell}\,\vartheta \mathfrak{I}_{\ell}}}
		\ i^{\sigma(\mathfrak{I}_{\ell})}\,C_{\mathfrak{I}_{\ell}}\,\vartheta(C_{\mathfrak{I}_{\ell}})
		\,e^{-H_{-}/k}\,e^{-\vartheta(H_{-})/k}}^{k}\nonumber\\
		&=& \sum_{\ell_{1},\ldots,\ell_{k}=0}^{L-1}
		i^{\sum_{i=1}^{k}\sigma(\mathfrak{I}_{\ell_{i}})}\,
		\mathfrak{c}_{\ell_{1},\ldots,\ell_{k}}\,Y_{\ell_{1},\ldots,\ell_{k}}\;.
	\eeq
In the second equality we have expanded the expression into a linear combination of $L^{k}$ terms  with coefficients 
	\be\label{eq:product_couplings}
		\mathfrak{c}_{\ell_{1},\ldots,\ell_{k}}
		= \frac{ 1}{k^{k}} \prod_{i=1}^{k}(-J_{\mathfrak{I}_{\ell_{i}}\,			\vartheta\mathfrak{I}_{\ell_{i}}})\,,
	\ee 
and with 
	\be
	\label{eq:main_identity=0}
		Y_{\ell_{1},\ldots,\ell_{k}}
		=
		C_{\mathfrak{I}_{\ell_{1}}}\vartheta(C_{\mathfrak{I}_{\ell_{1}}})\,e^{-H_{-}/k}
		\, e^{-\vartheta(H_{-})/k}\cdots C_{\mathfrak{I}_{\ell_{k}}}
		\vartheta(C_{\mathfrak{I}_{\ell_{k}}})\,e^{-H_{-}/k}\,e^{-\vartheta(H_{-})/k}\;.
	\ee 
Using this expansion,  \eqref{eq:Approx k} can be written  
	\be\label{eq:Expansion}
		A\,\vartheta(A)\lrp{e^{-H}}_{k}	
		= 
		\sum_{\ell_{1},\ldots,\ell_{k}=0}^{L-1}
		i^{\sum_{i=1}^{k}\sigma(\mathfrak{I}_{\ell_{i}})}\,
		\mathfrak{c}_{\ell_{1},\ldots,\ell_{k}}\,
		A\,\vartheta(A)\,Y_{\ell_{1},\ldots,\ell_{k}}\;.
	\ee

\begin{lem}\label{lem:Counting}
The trace $\Tr(A\,\vartheta(A)\,Y_{\ell_{1},\ldots,\ell_{k}})=0$ vanishes  unless  
	\be
	\label{eq:condition_sigma}
		\sum_{i=1}^{k}n(\mathfrak{I}_{\ell_{i}})=2\mathfrak{N}\;,
	\ee
is an even integer.  In this case, 
	\be
	\label{eq:Coupling Constants Positive}
		\sum_{i=1}^{k}\sigma(\mathfrak{I}_{\ell_{i}})=0\mod2\;,
		\qquad\text{and}\quad
		0 \leqslant 
		\mathfrak{c}_{\ell_{1},\ldots,\ell_{k}}\;.
	\ee
\end{lem}

\begin{proof}
In order to establish \eqref{eq:condition_sigma}, recall that we assume that the factor $A$ in $A\,\vartheta(A)\,Y_{\ell_{1},\ldots,\ell_{k}}$ is an element of $\mathfrak{A}_{-}^{\even}$.  Therefore we can expand it as a sum  of the form \eqref{eq:A_expansion}, with all the basis elements $M_{\beta}\in\mathfrak{A}_{-}^{\even}$.   As $H_{-}\in \mathfrak{A}_{-}^{\even}$, one can also expand each factor $e^{-H_{-}/k}$ as a sum of even basis elements $M_{\beta}\in\mathfrak{A}_{-}^{\even}$. Each {\em interaction term}, defined as a summand $C_{\mathfrak{I}_{\ell_{j}}}\vartheta(C_{\mathfrak{I}_{\ell_{j}}})$ in $H_{0}$, contains $n(\mathfrak{I}_{\ell_{j}})$ Majoranas in $\mathfrak{A}_{-}$ and an equal number in $\mathfrak{A}_{+}$.  
 
 We infer from Proposition \ref{prop:monomials} that the trace of $A\,\vartheta(A)\,Y_{\ell_{1},\ldots,\ell_{k}}$ vanishes unless each $c_{i}$ occurs in $A\,\vartheta(A)\,Y_{\ell_{1},\ldots,\ell_{k}}$ an even number of times. Consequently any $A\,\vartheta(A)\,Y_{\ell_{1},\ldots,\ell_{k}}$ with non-zero trace must have an even number of Majoranas in $\mathfrak{A}_{-}$.  In other words, the condition  \eqref{eq:condition_sigma} must hold. This ensures the number of odd $n(\mathfrak{I}_{\ell_{j}})$ is even. As $\sigma(\mathfrak{I}_{\ell_{j}})=n(\mathfrak{I}_{\ell_{j}})\mod2$, the sum of $\sigma(\mathfrak{I}_{\ell_{j}})$'s equals $0\mod 2$.

We next show that $0\leqslant \mathfrak{c}_{\ell_{1},\ldots,\ell_{k}}$. Suppose the interaction term $C_{\mathfrak{I}_{\ell_{j}}}\vartheta(C_{\mathfrak{I}_{\ell_{j}}})$  occurs as a factor in $A\,\vartheta(A)\,Y_{\ell_{1},\ldots,\ell_{k}}$ and has $\sigma(\mathfrak{I}_{\ell_{j}})=0$.  Then  the restriction on the coupling constants \eqref{eq:CouplingRestriction} means that $0\leqslant -J_{\mathfrak{I}_{\ell_{}j}\,\vartheta\mathfrak{I}_{\ell_{j}}}$. On the other hand, the condition \eqref{eq:Coupling Constants Positive} on $\sigma(\mathfrak{I}_{\ell_{j}})$ means that an even number of interaction terms in $A\,\vartheta(A)\,Y_{\ell_{1},\ldots,\ell_{k}}$ have $\sigma(\mathfrak{I}_{\ell_{j}})=1$.   From the restriction \eqref{eq:CouplingRestriction}, we infer that these couplings all have the same sign. Hence the product of the negative of these coupling constants is also positive.  Finally we use $0<J_{\varnothing\, \vartheta\varnothing}$ to complete the proof.
\end{proof}

\begin{lem}\label{lem:main_identity}
Assume relations \eqref{eq:condition_sigma}--\eqref{eq:Coupling Constants Positive}. Then the $Y_{\ell_{1},\ldots,\ell_{k}}$ in  \eqref{eq:main_identity=0} satisfy the identities
\be\label{eq:main_identity}
Y_{\ell_{1},\ldots,\ell_{k}}
=i^{-\sum_{i=1}^{k}\sigma(\mathfrak{I}_{\ell_{i}})}\,\ D_{{\ell_{1},\ldots,\ell_{k}}}\,\vartheta(D_{{\ell_{1},\ldots,\ell_{k}}})\,,
\ee
where 
\be\label{eq:C}
D_{{\ell_{1},\ldots,\ell_{k}}}=C_{\mathfrak{I}_{\ell_{1}}}\,e^{-H_{-}/k}\,C_{\mathfrak{I}_{\ell_{2}}}\,e^{-H_{-}/k}\cdots C_{\mathfrak{I}_{\ell_{k}}}\,e^{-H_{-}/k}\in\mathfrak{A}_{-}^{\even}\,.
\ee
\end{lem}
\begin{proof}
As  $e^{-H_{+}/k}=e^{-\vartheta (H_{-})/k}=\vartheta(e^{-H_{-}/k})$, the 
product $Y_{\ell_{1},\ldots,\ell_{k}}$ in \eqref{eq:main_identity=0}  differs from the product $D_{\ell_{1},\ldots,\ell_{k}}\,\vartheta(D_{\ell_{1},\ldots,\ell_{k}})$, only in the order of its factors.  
In order to transform from one product into the other, we need to move all the Majorana operators of $Y_{\ell_{1},\ldots,\ell_{k}}$ that are localized in $\mathfrak{A}_{-}$ to the left, and all operators of $Y_{\ell_{1},\ldots,\ell_{k}}$ in $\mathfrak{A}_{+}$ to the right.
We move each operator $c_{j}$ as far as possible to the left, without permuting the order of any operator in $\mathfrak{A}_{-}$. As $H_{+}\in\mathfrak{A}_{+}^{\even}$, it commutes with each $c_{j}\in\mathfrak{A}_{-}$.  Likewise  $H_{-}\in\mathfrak{A}_{-}^{\even}$, it commutes with each $c_{j}\in\mathfrak{A}_{+}$.
This procedure neither changes any of the exponentials $e^{-H_{\pm}/k}$. It gives rise to a minus sign only each time we permute a $c_{j}$ in an interaction term to the left past an operator $\vartheta(c_{j'})$ in another interaction term. 

We count the minus signs that occur from permuting the $c$'s in the interaction terms. In order to simplify notation, let $n_{\ell_{i}}=n(\mathfrak{I}_{\ell_{i}})$.  The term $C_{\mathfrak{I}_{\ell_{1}}}\vartheta(C_{\mathfrak{I}_{\ell_{1}}})$ contributes no minus sign. The term $C_{\mathfrak{I}_{\ell_{2}}}\vartheta(C_{\mathfrak{I}_{\ell_{2}}})$ contributes $n_{\ell_{2}}n_{\ell_{1}}$ minus signs. The term $C_{\mathfrak{I}_{\ell_{3}}}\vartheta(C_{\mathfrak{I}_{\ell_{3}}})$ contributes $n_{\ell_{3}}(n_{\ell_{1}}+n_{\ell_{2}})$ minus signs.  The term $C_{\mathfrak{I}_{\ell_{4}}}\vartheta(C_{\mathfrak{I}_{\ell_{4}}})$ contributes $n_{\ell_{4}}(n_{\ell_{1}}+n_{\ell_{2}}+n_{\ell_{3}})$ minus signs, and so on.  Finally, the term $$C_{\mathfrak{I}_{\ell_{k}}}\vartheta(C_{\mathfrak{I}_{\ell_{k}}})$$ contributes $n_{\ell_{k}}\sum_{i=1}^{k-1}n_{\ell_{i}}$ minus signs.  Adding these numbers, one obtains a total number of minus signs equal to
	\be
	\label{eq:minus_signs}
		\frac{1}{2}\sum_{i,i'=1}^{k} n_{\ell_{i}}\, n_{\ell_{i'}}
			-\frac{1}{2}\sum_{i=1}^{k}  n_{\ell_{i}}^2
		= \frac{1}{2}\lrp{\sum_{i=1}^{k} 
			n_{\ell_{i}}}^2
			- \frac{1}{2} \sum_{i=1}^{k}  n_{\ell_{i}}^2
		= 2\mathfrak{N}^2- \frac{1}{2} \sum_{i=1}^{k}  n_{\ell_{i}}^2\,.
	\ee
Here $\mathfrak{N}$ is defined in \eqref{eq:condition_sigma}.
We infer that
	\be 
	\label{Sign from Permutation}
		\lrp{2\mathfrak{N}^2- \frac{1}{2} \sum_{i=1}^{k}   n_{\ell_{i}}^2} \mod2
		=  -\frac{1}{2} \sum_{i=1}^{k}  n_{\ell_{i}}^2 \mod2\;.
	\ee

The overall  sign arising from the permutation of the $c$'s in going from \eqref{eq:main_identity=0} to \eqref{eq:main_identity} is $(-1)$ raised to the power \eqref{Sign from Permutation}.  This is   
	\be
		(-1)^{-\frac{1}{2}\sum_{i=1}^{k}n_{\ell_{i}}^2}
		=i^{-\sum_{i=1}^{k}n_{\ell_{i}}^2}
		=i^{-\sum_{i=1}^{k}(n_{\ell_{i}}\mod2)}
		=i^{-\sum_{i=1}^{k}\sigma_{\ell_{i}}}\,.
	\ee
In the second equality we use an identity for natural numbers $n$, namely 
	\be
		n^2\mod4=n\mod 2\;.
	\ee
In the final equality we use  the definition $\sigma_{\ell_{i}}=n_{\ell_{i}}\mod2$.
\end{proof}

\noindent{\it Completion of the proof of Theorem \ref{prop:reflection_positivity}.}  
In case $\Tr(A\,\vartheta(A)\,Y_{\ell_{1},\ldots,\ell_{k}})\neq0$, we infer from\eqref{eq:Expansion} along with Lemmas \ref{lem:Counting} and  \ref{lem:main_identity} and the fact that $\vartheta(A)$ commutes with $D_{\ell_{1},\ldots,\ell_{k}}$  that	
	\be\label{eq:ExpansionResult}
		\Tr\lrp{A\,\vartheta(A)\,e^{-H}}
		=\lim_{k\to\infty} 
		\sum_{\ell_{1},\ldots,\ell_{k}=0}^{L-1}
		\mathfrak{c}_{\ell_{1},\ldots,\ell_{k}}\,
		\Tr\lrp{AD_{\ell_{1},\ldots,\ell_{k}}
		\vartheta\lrp{AD_{\ell_{1},\ldots,\ell_{k}}}}\;.
	\ee
Notice that the factors of $i$ in \eqref{eq:Expansion} cancel against the factors of $i$ in \eqref{eq:main_identity}, so there are no factors of $i$ in \eqref{eq:ExpansionResult}. In the last statement of Lemma  \ref{lem:Counting}, we have established that $0\leqslant \mathfrak{c}_{\ell_{1},\ldots,\ell_{k}}$.  And from Proposition~\ref{prop:positivity}, we infer that   $0 \leqslant \Tr\lrp{AD_{\ell_{1},\ldots,\ell_{k}}\vartheta\lrp{AD_{\ell_{1},\ldots,\ell_{k}}}}$.  Thus \eqref{eq:ExpansionResult}  is a sum of positive terms. 
This completes the proof in the case that $A\in\mathfrak{A}_{-}^{\even}$.  

The remaining case is $A\in\mathfrak{A}_{+}^{\even}$.  Then one has  $A=\vartheta(\widetilde{A})$ with $\widetilde{A}\in\mathfrak{A}_{-}^{\even}$. As $A$ commutes with $\vartheta(A)$, we infer that $A\,\vartheta(A)=\widetilde{A} \,\vartheta(\widetilde {A})$, and $\Tr\lrp{A\,\vartheta(A)\,e^{-H}}=\Tr\lrp{\widetilde{A} \,\vartheta(\widetilde {A})\,e^{-H}}\geqslant0$ as a consequence of the case already established.
\end{proof}

\subsection{Reflection-Positive Inner Product}
Let us introduce the modified pre-inner product on $\mathfrak{A}_{\pm}^{\even}$ defined by the functional \eqref{eq:rp}. Let
\be\label{eq:inner_product_rp}
\lra{A,B}_{\rp}=\Tr(A\,\vartheta(B)\,e^{-H})\,.
\ee
Denote the corresponding semi-norm by $\Vert A \Vert_{\rp}$.

The theorem shows that one has an elementary reflection positivity bound, arising from the Schwarz inequality.  Also $\vartheta$ acts as anti-unitary transformation on the Hilbert space $\mathfrak{A}_{\pm}^{\even}$ with inner product \eqref{eq:inner_product_rp}.
\begin{cor}
For $A,B\in\mathfrak{A}_{\pm}^{\even}$,	one has
\be
\abs{\lra{A,B}_{\rp}\,}
\leqslant \norm{A}_{\rp}\,\norm{B}_{\rp}\;,
\ee
and
\be
\lra{A,B}_{\rp}
=
\lra{\vartheta(B),\vartheta(A)}_{\rp}\,,
\qquad\text{so}\quad
\norm{\vartheta(A)}_{\rp} 
= \norm{A}_{\rp}\;. 
\ee
\end{cor}

\section{Relation to Spin Systems}
It is well-known that the ferromagnetic Ising model is reflection-positive, but the quantum Heisenberg model is not reflection-positive \cite{FILS}. We can also infer these facts from the point of view of Majoranas. 

One can consider the infinitesimal rotation matrices in the $(\alpha,\beta)$-plane,  $\Sigma^{\alpha\beta}=\frac{-i}{2} \comm{\gamma^{\alpha}}{\gamma^{\beta}}$, with  $\gamma^{\alpha}$ the Euclidean  Dirac matrices on  $4$-space with coordinate labels $\alpha,\beta\in\{0,x,y,z\}$. 
In the notation sometimes  used in condensed-matter physics, one assigns Dirac matrices $\gamma_{j}^{\alpha}$ as four Majoranas $ c_{j},b_{j}^{x}, b_{j}^{y}, b_{j}^{z}$  at each lattice site. We use a real representation for $b_{j}^{x}$ and $b_{j}^{z}$, and an imaginary representation for $b_{j}^{y}$ and $c_{j}$.  One could also use a real representation for $b^{y}_{j}$ and $c_{j}$, and an imaginary representation for $b_{j}^{x}$ and $b_{j}^{z}$.

Then the three $(0,\alpha)$ planes yield $\Sigma^{0\alpha}_{j}=\sigma_{j}^{\alpha}=i\, b_{j}^{\alpha}c_{j}$.  They agree with the Pauli matrices when projected to one chiral copy, namely to the subspace of the Hilbert space $\calh$ of the Majoranas, on which each of the mutually commuting operators $\gamma_{j}^{5}=b_{j}^{x}b_{j}^{y}b_{j}^{z}c_{j}$  has the eigenvalue $+1$. Note that each $\gamma_{j}$ commutes with  all the $\vec \sigma_{j'}$.  With these choices, the $\sigma^{x, z}_{j}$ are real, while $\sigma^{y}_{j}$ is imaginary.\footnote{ These three operators correspond to half of the generators $\Sigma_{j}^{\alpha\beta}$, and we use this representation. The other three generators $\Sigma_{j}^{\alpha\beta}$  for $\alpha,\beta\neq0$  act the same on both chiral copies, and as they are isomorphic on each copy they yield an alternative representation $\sigma^{x}_{j}=-i b^{y}_{j}\,b^{z}_{j}$, etc., which is also sometimes used in the condensed-matter literature.}  

For a reflection across a nearest-neighbor bond $(ij)$, a ferromagnetic Ising interaction term is 
	\be
		-\sigma_{i}^{z}\sigma_{j}^{z}
		= b_{i}^{z}c_{i} b_{j}^{z} c_{j}
		= - b_{i}^{z}c_{i} \,\vartheta\lrp{b_{i}^{z}c_{i}}\;.
	\ee
This satisfies condition \eqref{eq:CouplingRestriction} with $k=2$ and $\sigma=0$. Similarly, the quantum ``rotator'' Hamiltonian has an interaction term,  
\be
		- \sigma_{i}^{x} \,\sigma_{j}^{x}
		- \sigma_{i}^{z} \,\sigma_{j}^{z}
		=  -b_{i}^{x}c_{i} \,\vartheta\lrp{b_{i}^{x}c_{i}}
		-  b_{i}^{z}c_{i} \,\vartheta\lrp{b_{i}^{z}c_{i}}\;.
	\ee
This also satisfies condition \eqref{eq:CouplingRestriction}, and so is reflection-positive.   The corresponding quantum Heisenberg  interaction term is 
	\be
		-\vec \sigma_{i} \cdot \vec \sigma_{j}
		= - \sigma_{i}^{x} \,\sigma_{j}^{x} 
		-\sigma_{i}^{y} \,\sigma_{j}^{y} 
		- \sigma_{i}^{z} \,\sigma_{j}^{z}
		=
		 -b_{i}^{x}c_{i} \,\vartheta\lrp{b_{i}^{x}c_{i}}
		+ b_{i}^{y}c_{i} \,\vartheta\lrp{b_{i}^{y}c_{i}}
		-  b_{i}^{z}c_{i} \,\vartheta\lrp{b_{i}^{z}c_{i}}\;.
	\ee
This does not satisfy \eqref{eq:CouplingRestriction}, since one of the interaction coefficients of the term $ b_{i}^{y}c_{i} \,\vartheta\lrp{b_{i}^{y}c_{i}}$ arising from $-\sigma_{i}^{y} \,\sigma_{j}^{y} $ is positive.   

\setcounter{equation}{0}  \section{Reflection Bounds}\label{sect:Reflection_Bounds}
The use of reflection bounds and their iteration has many applications, both in statistical physics and quantum field theory. Here we study some bounds which follow from the results of Section \ref{sec:RP}, that we apply in \cite{CJLP}.

Let us introduce two pre-inner products $\lra{\ \cdot\,,\cdot \ }_{\rp\pm}$ on the algebras $\mathfrak{A}_{\pm}^{\even}$, corresponding to two reflection symmetric Hamiltonians. Let 
\be
\lra{A,B}_{\rp-}=\Tr(A\,\vartheta(B)\,e^{-H})\,,\quad\text{for}\quad H=H_{-}+H_{0}+\vartheta(H_{-})\,.
\ee
Similarly define
\be
\lra{A,B}_{\rp+}=\Tr(A\,\vartheta(B)\,e^{-H})\,,\quad\text{for}\quad H=\vartheta(H_{+})+H_{0}+H_{+}\,.
\ee
As previously, one can define inner products on equivalence classes, yielding norms $\Vert \ \cdot\  \Vert$.

\begin{prop}[\bf RP-Bounds]
Let $H=H_{-}+H_{0}+H_{+}$ with $H_{\pm}\in\mathfrak{A}_{\pm}^{\even}$ and $H_{0}$ of the form \eqref{eq:Hamiltonian_interaction}. Then
\be\label{eq:bound_1}
\abs{\Tr(A\,\vartheta(B)\,e^{-H})}\leqslant \Vert A\Vert_{\rp-}\,\Vert B\Vert_{\rp+}\,,\qquad\text{for}\quad A,B\in\mathfrak{A}_{-}^{\even}\,.
\ee
Also
\be\label{eq:bound_2}
\abs{\Tr(A\,\vartheta(B)\,e^{-H})}\leqslant \Vert A\Vert_{\rp+}\,\Vert B\Vert_{\rp-}\,,\qquad\text{for}\quad A,B\in\mathfrak{A}_{+}^{\even}\,.
\ee
In particular for $A=B=I$, 
\be
\Tr(e^{-H})\leqslant \Tr(e^{-(H_{-}+H_{0}+\vartheta(H_{-}))})^{1/2}\,\Tr(e^{-(\vartheta(H_{+})+H_{0}+H_{+})})^{1/2}\,.
\ee
\end{prop}

\begin{proof}
The proof of \eqref{eq:bound_1} follows the proof of Theorem \ref{prop:reflection_positivity}. Use the expression \eqref{eq:Exp_minusH_k} to write $A\,\vartheta(B)\,\lrp{e^{-H}}_{k}$, which converges to $A\,\vartheta(B)\,e^{-H}$ as $k\to\infty$, namely 
	\beq\label{eq:ExpansionResult-2}
		\Tr\lrp{A\,\vartheta(B)\lrp{e^{-H}}_{k}}
		&=&
		\sum_{\ell_{1},\ldots,\ell_{k}=0}^{L-1}
		\mathfrak{c}_{\ell_{1},\ldots,\ell_{k}}\,
		\Tr\lrp{AD^{-}_{\ell_{1},\ldots,\ell_{k}}
		\vartheta\lrp{BD^{+}_{\ell_{1},\ldots,\ell_{k}}}}\nonumber\\
		&=&
		\sum_{\ell_{1},\ldots,\ell_{k}=0}^{L-1}
		\mathfrak{c}_{\ell_{1},\ldots,\ell_{k}}
		\lra{AD_{\ell_{1},\ldots,\ell_{k}}^{-},
		BD_{\ell_{1},\ldots,\ell_{k}}^{+}}_{\rp}\,.\nonumber\\
	\eeq
The form $\lra{\ \cdot\ ,\ \cdot\ }_{\rp}$ in \eqref{eq:ExpansionResult-2} is defined in \eqref{eq:inner_product_1}.  The difference is that now the terms contain $\vartheta(B)$ in place of $\vartheta(A)$, and  $D^{\pm}_{\ell_{1},\ldots,\ell_{k}}$ depends on  $H_{\pm}$. Thus  the  constants $\mathfrak{c}_{\ell_{1},\ldots,\ell_{k}}$ are given by \eqref{eq:product_couplings}, the matrices
$D_{\ell_{1},\ldots,\ell_{k}}^{-}\in\mathfrak{A}_{-}^{\even}$ are given by \eqref{eq:C},
and
	\be\label{eq:Cplus}
		\vartheta(D_{\ell_{1},\ldots,\ell_{k}}^{+})
		=\vartheta(C_{\mathfrak{I}_{\ell_{1}}})e^{-H_{+}/k}
			\vartheta(C_{\mathfrak{I}_{\ell_{2}}})e^{-H_{+}/k}\cdots 
			\vartheta(C_{\mathfrak{I}_{\ell_{k}}})
			e^{-H_{+}/k}\in\mathfrak{A}_{+}^{\even}\,.
	\ee
Lemma \ref{lem:Counting} depends only on the form of $H_{0}$ and the fact that $H_{\pm}\in\mathfrak{A}_{\pm}^{\even}$. Thus the lemma applies in this case as well. With these substitutions, the proof of Lemma \ref{lem:main_identity} also applies.  
%
%
%

To establish \eqref{eq:bound_1}, note that the product of couplings $\mathfrak{c}_{\ell_{1},\ldots,\ell_{k}}$ defined in \eqref{eq:product_couplings} are independent of $A$ and $B$, so as before we infer from Lemma \ref{lem:Counting} that $\mathfrak{c}_{\ell_{1},\ldots,\ell_{k}}\geqslant0$ whenever 
	$
		\lra{AD_{\ell_{1},\ldots,\ell_{k}}^{-},
		BD_{\ell_{1},\ldots,\ell_{k}}^{+}}_{\rp}
		\neq0
	$. 
Use the Schwarz inequality for $\lra{\ \cdot\,, \cdot\ }_{\rp}$ and the positivity of $\mathfrak{c}_{\ell_{1},\ldots,\ell_{k}}$ to obtain
\beq
	\abs{\Tr\lrp{A\,\vartheta(B)\,e^{-H}}}
	&=&
	\abs{\lim_{k\rightarrow\infty}
	\sum_{\ell_{1},\ldots,\ell_{k}=0}^{L-1}
	\mathfrak{c}_{\ell_{1},\ldots,\ell_{k}}\lra{AD_{\ell_{1},
	\ldots,\ell_{k}}^{-},BD_{\ell_{1},\ldots,\ell_{k}}^{+}}_{\rp}}
	\nonumber\\
	&\leqslant& \lim_{k\rightarrow\infty}
	\sum_{\ell_{1},\ldots,\ell_{k}=0}^{L-1}
	\mathfrak{c}_{\ell_{1},\ldots,\ell_{k}}^{1/2}\,\lra{AD_{\ell_{1},\ldots,\ell_{k}}^{-},AD_{\ell_{1},\ldots,\ell_{k}}^{-}}_{\rp}^{1/2}\nonumber\\
	&& \qquad \times\ \mathfrak{c}_{\ell_{1},\ldots,\ell_{k}}^{1/2}
	\,\lra{BD_{\ell_{1},\ldots,\ell_{k}}^{+},BD_{\ell_{1},
	\ldots,\ell_{k}}^{+}}_{\rp}^{1/2}\nonumber\\
	&\leqslant&\lim_{k\rightarrow\infty}
	\lrp{\sum_{\ell_{1},\ldots,\ell_{k}=0}^{L-1}
		\mathfrak{c}_{\ell_{1},\ldots,\ell_{k}}\lra{AD_{\ell_{1},
	\ldots,\ell_{k}}^{-},
		AD_{\ell_{1},\ldots,\ell_{k}}^{-}}_{\rp}}^{1/2}\nonumber\\
	&& \quad \times
	\lrp{
	\sum_{\ell_{1},\ldots,\ell_{k}=0}^{L-1}
	\mathfrak{c}_{\ell_{1},\ldots,\ell_{k}}\,
	\lra{BD_{\ell_{1},\ldots,\ell_{k}}^{+},BD_{\ell_{1},\ldots,\ell_{k}}^{+}}_{\rp}}^{1/2}\nonumber\\
	&=&\lra{A,A}_{\rp-}^{1/2}\ \lra{B,B}_{\rp+}^{1/2}
	=\Vert A\Vert_{\rp-}\,\Vert B\Vert_{\rp+}\,.
\eeq
This completes the proof of relation \eqref{eq:bound_1}.

When $A,B\in\mathfrak{A}_{+}^{\even}$, substitute in the left-hand side of \eqref{eq:bound_2} $A=\vartheta(\widetilde{A})$ and $B=\vartheta(\widetilde{B})$ with $\widetilde{A}, \widetilde{B}\in\mathfrak{A}_{-}^{\even}$. Since $A$ and $B$ commute with $\vartheta(A)$ and $\vartheta(B)$,
	\be
		\abs{\Tr(A\,\vartheta(B)\,e^{-H})}=							\abs{\Tr(\widetilde{B}\,\vartheta(\widetilde{A})\,e^{-H})}\,.
	\ee
Replacing $H_{-}$ by $\vartheta(H_{+})$ and $\vartheta(H_{-})$ by $H_{+}$ in the bound \eqref{eq:bound_1} completes the proof of \eqref{eq:bound_2}.
\end{proof}

\setcounter{equation}{0}  \section{Acknowledgement}
Arthur Jaffe wishes to thank Daniel Loss for his warm hospitality at the University of Basel, Department of Physics, where much of this work was carried out, and for creating the stimulating atmosphere that made this work possible.  We are also grateful for comments by J\"urg Fr\"ohlich and by Stefano Chesi.
This work was supported by the Swiss NSF, NCCR QSIT, NCCR Nanoscience, and the Pauli Center ETHZ.

\end{document}